\documentclass[10pt,journal,compsoc]{IEEEtran}
\ifCLASSOPTIONcompsoc
  \usepackage[nocompress]{cite}
\else
  \usepackage{cite}
\fi
\ifCLASSINFOpdf
\else
\fi

\usepackage[cmtip,all]{xy}
\usepackage[english]{babel}
\usepackage{amsmath,amssymb}

\usepackage{times}
\usepackage{relsize}
\usepackage{graphicx}
\usepackage{url}
\usepackage{nicefrac}
\usepackage{booktabs}
\usepackage{multirow}
\usepackage{mparhack}
\usepackage{subfigure}
\usepackage{caption}
\usepackage{amsthm}
\usepackage{bm}
\usepackage{algorithm}
\usepackage[noend]{algorithmic}

\usepackage{array}
\usepackage{eqparbox}
\usepackage{mdwmath}

\usepackage{epsfig}
\usepackage{epstopdf}
\usepackage{xcolor}

\newtheorem{theorem}{Theorem}

\newcounter{exam}
\newcounter{rema}
\DeclareFontFamily{OT1}{pzc}{}
\DeclareFontShape{OT1}{pzc}{m}{it}%
              {<-> s * [1.1] pzcmi7t}{}
\DeclareMathAlphabet{\mathpzc}{OT1}{pzc}%
                                 {m}{it}

\usepackage{accents}

\begin{document}
\captionsetup[figure]{name={Fig.},font=footnotesize,labelsep=period}
\title{Optimal Content Caching and Recommendation with Age of Information}

\author{Ghafour~Ahani and 
        Di~Yuan,~\IEEEmembership{Senior Member,~IEEE}

\IEEEcompsocitemizethanks{\IEEEcompsocthanksitem G. Ahani and D. Yuan are with the Department
of Information Technology, Uppsala University, 751 05 Uppsala, Sweden (e-mails: \{ghafour.ahani, di.yuan\}@it.uu.se).\protect\\
}
}
\IEEEtitleabstractindextext{%
\begin{abstract}
Content caching at the network edge has been considered an effective way of mitigating backhaul load and improving user experience. 
Caching efficiency can be enhanced by content recommendation and by keeping the information fresh.
To the best of our knowledge, there is no work that
jointly takes into account these aspects.
By content recommendation, a requested content that is not in the cache
can be alternatively satisfied by a related cached content recommended by the system. 
Information freshness can be quantified by age of
information (AoI). We address, optimal scheduling of cache updates
for a time-slotted system accounting for content recommendation and AoI.
For each content, there are requests that need to be satisfied, and there is a cost function capturing the freshness of information.
We present the following contributions.
First, we prove that the problem is NP-hard. 
Second, we derive an integer linear formulation, by which the optimal solution can be obtained for small-scale
scenarios.
Third, we develop an algorithm based on Lagrangian decomposition. Fourth, we develop efficient algorithms for solving the resulting subproblems.
Our algorithm computes a bound that can be used to evaluate the performance of any suboptimal solution. 
Finally, we conduct simulations to show the effectiveness of our algorithm in comparison to a greedy schedule.
\end{abstract}

\begin{IEEEkeywords}
Age of information, caching, content recommendation, Scheduling
\end{IEEEkeywords}}
\maketitle
\IEEEpeerreviewmaketitle

\section{Introduction}
Content caching at network edge, such as base stations, is a promising solution to deal with the explosively increasing traffic demand and to improve user experience~\cite{2016_liu}.
This approach is beneficial for both the users and the network
operators as the former can access the content
at a reduced latency, and latter can alleviate the load on backhaul links. The performance of
edge caching, however, can be further improved by utilizing content recommendation and optimizing information freshness.


Originally, recommender systems have been used for presenting
content items that best match user interests and preferences.
In fact, the reports in \cite{2016_carlos,2010_Zhou} show that $80\%$ of requests on content distribution platforms are due to content recommendations. Recently, a number of studies have proposed to utilize content recommendation for improving caching efficiency.
In~\cite{2019_chat,2017_chat}, recommendation is utilized to 
steer user requests toward the contents that are both stored in the cache and of interest to users. 
More recently, content recommendation is employed to satisfy content requests using alternative and related contents.
Namely, instead of the initially requested content that is
absent from the cache, some other related contents are recommended~\cite{2018_sermpezis,2018_song,2020_Costantini}.
This approach is of interest to many applications such as video and image retrieval, and entertainment-based ones~\cite{2018_sermpezis}.

Another important aspect that arises naturally in the context of content caching 
is the freshness of information~\cite{2017_yates}. As cached contents may become
obsolete with time, we need to also account for updating the content items.
Information freshness is quantified by age of information (AoI) which is defined as the amount of time elapsed with respect to the time stamp
of the information in the most recent update~\cite{2012_kaul}. The AoI grows linearly between two successive updates.

In this study, we address optimal  scheduling  for updating the  cache for  a  time-slotted  system where  content  recommendation  and  AoI are jointly accounted  for. The cache has a capacity limit, and the content items vary by size. Moreover, updating the cache in a time
slot is subject to a network capacity limit.
 For a content request, if the content is available in the cache, the request is served using the stored content. Otherwise,
a set of related and cached contents will be recommended. If one of the recommended contents
is accepted, then the request will be again served from the cache with the accepted content. If not, the request will be served by the remote server with a higher cost.
It is worth noting that incentive mechanisms may be utilized to motivate users to accept the recommended contents (e.g., zero-rating services)~\cite{2018_sermpezis}.
For each content item, there is a cost function that is monotonically increasing in the AoI. Thus, caching a content with higher AoI
results in a higher cost.

The optimization decision consists of, the selection of the content items
for updating the cache, and a recommendation set for each non-cached content. The objective is to find the schedule
minimizing the total cost over the scheduling horizon.

Our work consists in the following contributions for the outlined
cache optimization problem with recommendation and AoI (COPRA):

\begin{itemize}
    \item We rigorously prove the NP-hardness of the problem, even when contents are of uniform size, based on a reduction from 3-satisfiability (3-SAT).
    \item We derive an integer linear programming (ILP) formulation for the problem in its general form, enabling the use of
        general-purpose optimization solvers to approach the problem. This is particularly useful for solving small-scale problem instances
        to optimality, and to enable to accurately evaluate low-complexity though sub-optimal algorithms.
    \item For problem solving, we apply Lagrangian decomposition to the ILP, allowing for decomposing the problem
    into two subproblems, each with special structures. The first subproblem itself further decomposes into smaller problems, each of which can be mapped to finding a shortest path in a graph. 
    The second subproblem also decomposes to smaller problems. However, the problem size remains exponential, and therefore we propose column generation for gaining optimality.
     Moreover, we demonstrate that the pricing problem of column generation can be solved via dynamic programming (DP). It is also worth noting that, decomposition enables parallel computation. In addition, our algorithm computes a lower bound (LBD) that can be
    used to evaluate the quality of any given solution.
    \item Finally, we conduct extensive simulations to evaluate the performance of our algorithm by comparing its solution to global optimum for small-scale scenarios,  and to the LBD otherwise. The evaluations show that
    our algorithm provides solutions within $8\%$ of global optimality.
    
\end{itemize}

\section{Related Work}
To the best of our knowledge there is no work that jointly study content caching, recommendation, and AoI.
In the following, we first review the works that have studied content caching and AoI, and then those on caching and recommendation.

The works in~\cite{2017_yates,2018_zhangshan,2107_kam,2020_ahani_icc,2020_tang,2020_ahani_tmc,2020_bastopcu,2021_bastopcu} have studied content caching when AoI is accounted for.
The general problem setup in these works is what contents to cache and when to update them with an objective function based on AoI.
In~\cite{2017_yates,2018_zhangshan} the objective is minimizing the expected AoI when inter-update intervals of each item or
the total number of updates are known. In~\cite{2107_kam,2020_ahani_icc,2020_tang} content caching is studied where both popularity
and AoI are considered. In~\cite{2107_kam}, cache miss is minimized, and in \cite{2020_ahani_icc} the load of backhaul link is minimized via balancing
the AoI and cache updates. In~\cite{2020_tang}, partially updating a content, which depends on the type of content and its AoI, is enough to completely update the content. In~\cite{2020_ahani_tmc}, the overall utility of a cache defined based on AoI of contents is maximized, subject to limited cache and backhaul link capacities. 
For a given origin, a set of users, and a (set of) cache between them, the AoI at the cache(s) and users is analyzed in~\cite{2021_bastopcu}.
As an extension of~\cite{2021_bastopcu}, the work in \cite{2020_bastopcu} studied the trade-off between obtaining a 
content from the origin with longer transmission time and from the cache with higher AoI. A recent survey of AoI can be found in \cite{2021_yates}.

In general, the works that studied content caching and recommendation can be classified into two categories.
In the first category, content recommendation is utilized to shape the requests and steer the content 
demand toward the contents that are both stored in the cache and interesting to the users~
\cite{2017_chat,2019_chat,2020_Savvas,2020_tsigkari,2020_fu,2021_fu,2019_guo,2021_gian,2019_dong}.
In~\cite{2017_chat,2019_chat} a preference ``distortion'' tolerance measure is used to quantify how much
the engineered recommendations distort the original user content preferences. 
In~\cite{2020_Savvas}, an experiment is conducted to
demonstrate the effect of content recommendation on caching efficiency in practice.
In~\cite{2020_tsigkari}, the objective is to maximize both the quality of recommendation and streaming rate, and the authors proposed a polynomial-time algorithm with approximation guarantee. In~\cite{2020_fu,2021_fu}, caching and recommendation decisions are optimized based on the preference distribution of individual users.
In~\cite{2019_guo}, content caching and recommendation are optimized taking into account the temporal-spatial variability of user requests.
In~\cite{2021_gian}, the authors study the fairness issues of recommendation where some contents get more visible than others by recommendation.
In~\cite{2019_dong}, reinforcement learning is utilized for learning user behavior and optimizing caching and recommendation.

In the second category of studies,
recommendation is utilized to satisfy a request when the requested content is not available in the cache, by
recommending some other cached and related contents~\cite{2018_sermpezis,2020_Costantini,2020_Costantini1,2018_song,2020_Garetto,2020_zhou}. The idea of recommending related contents in case of a cache miss is formally introduced in \cite{2018_sermpezis} where the authors referred to the scenario as ``soft cache hit''.
In this reference, the authors illustrate how ``soft cache hit'' is able to improve the caching performance. 
They also consider a caching problem with the objective of maximizing the cache hit rate where all contents in the cache can be recommended. 
Using the submodularity property of the objective function, they propose algorithms with performance guarantee. Later,
in \cite{2020_Costantini}, the authors consider a more realistic system model in which only a limited number of contents can be recommended.
Then, they propose a polynomial-time algorithm based on first solving the caching problem, and then finding the recommendations sets.
In~\cite{2020_Costantini1}, the authors model the relation among contents as a graph, and then studied the characteristics of this graph to
predict whether it is worth to find the optimal solution or a low complexity heuristic will be sufficient. 
In~\cite{2020_Garetto}, the authors try to find the best caching  policy for a sequence of requests where recommendation is accounted for.
In~\cite{2020_zhou} a multi-hop cache network
is studied where soft cache hit is allowed in one of the caches along the path to the end node that stores the initially requested content.

The closest works to our study are~\cite{2018_sermpezis,2020_Costantini} in the sense that they also considered soft cache hits. 
However, there are significant differences. To the best our knowledge, it is novel that that caching decision, 
content recommendation, and information freshness are jointly
optimized. Moreover, in our work we account for cache update costs, 
as well as the capacities of cache and backhaul links.

\section{System Scenario and Complexity Analysis}

\subsection{System Scenario}\label{System_Scenario}
 The system scenario consists of a content server, a base station (BS), 
 and a set of content items $\mathcal{I}=\{1,2,\dots,I\}$. 
 The  server has all the contents, and the BS is equipped with a cache of capacity $S$. 
 The BS is connected to the content server with a communication link 
 of capacity $L$ via which the cache contents can be updated.
 The size of content item $i\in \mathcal{I}$ is denoted by $s_i$.

We consider a time-slotted system with a time period of $T$ time slots, 
denoted by $\mathcal{T}=\{1,2,\dots,T\}$. 
At the beginning of each time slot, the contents of the cache are subject to updates.
Namely, some stored contents may be removed from the cache, 
some new contents may be added to the cache by downloading from the server,
and some existing contents may be refreshed.

The AoI of an item in the cache is the time difference between the current slot
and the time slot in which the item was most recently downloaded to the cache.
Each time an item is downloaded to the cache, the item’s AoI is zero, i.e., maximum information freshness.
The AoI then increases by one for each time slot, until the next update.
In other words, the AoI of any cached content item is linear in time, if the content is not refreshed.
For content $i\in\mathcal{I}$, the relevant AoI has a limit $A_i$. 
The content is considered obsolete if the AoI exceeds $A_i$.
Hence, a cached content $i$ in time slot $t$ can take one of the AoIs in $\mathcal{A}_{ti}=\{0,...,\min(A_i,t-1)\}$.
The cost associated with content item $i$ with AoI $a$ in time slot $t$ is characterized by a cost function $f_{tia}$ that is monotonically increasing
in AoI~$a$.

For a request of content~$i$, if the content is stored in the cache and the AoI is no more than $A_i$, the request is satisfied from the cache.
Otherwise, a set of related cached contents, hereinafter referred to as a recommendation set, is recommended to the user. 
If the user accepts any element of the recommendation set, the request is satisfied by the cache.
If not, the request needs to be satisfied from the server. Note that since a user
may not be interested in getting a long list of recommended contents, we limit
the size of recommendation set to be at most $N$~\cite{2018_liu,2020_Costantini}.
Denote by $\mathcal{R}_i=\{1,2,...,R_i\}$ the index set of all contents related to content $i$. 
This set can be determined from past statistics and/or learning algorithms~\cite{2018_sermpezis}.
Obviously, the index set of any recommendation set for content $i$ is a subset of $\mathcal{R}_i$.
Note that the recommendation set may change from a time slot to another.

Denote by $h_{ti}$ the number of requests for content $i\in \mathcal{I}$ in time slot $t\in \mathcal{T}$. 
The value of $h_{ti}$ can be estimated via recent requests of the contents, popularity of the contents, and/or machine learning algorithms~\cite{2018_sermpezis,2018_zhang}.
In this study, for the ease of exposition, we consider the total number of requests for a content instead of individual user requests. Similarly, 
the acceptance probability of a content does not vary from a user to another.
Note that
individual user requests and acceptance probability can be easily accommodated in our formulations and algorithms. Denote by $c_s$ and $c_b$ the costs for downloading one unit of data from the server and from the cache to a user, respectively. Downloading cost from server to cache is $c_s-c_b$. Intuitively, $c_s>c_b$ to encourage downloading from the cache.

The cache optimization problem with recommendation and AoI, or COPRA in short,
is to determine which content items to store, update, and recommend
in each time slot, such that the total cost of content requests over time horizon
$1,2,...,T$ is minimized, subject to cache and backhaul link capacities.

\subsection{Cost Model}
Denote by $x_{tia}$ a binary optimization variable that equals one if and only if content $i$ with AoI $a$ is stored in the cache at time slot $t$. Hence, $x_{ti0}=1$ means that the content $i$ at time slot $t$ is just downloaded from the server to the cache with AoI zero. Then, the overall downloading cost is shown in \eqref{delta1}. In \eqref{delta1}, the first term is the downloading cost from server to the cache due to cache updates and the second term is the downloading cost of requests that are delivered using cached contents.

\begin{equation}\label{delta1}
\Delta_{download}=\sum_{t \in \mathcal{T}}\sum_{i \in \mathcal{I}}\Big((c_s-c_b)s_ix_{ti0}+\sum_{a\in\mathcal{A}_{ti}}c_bs_ih_{ti}f_{tia}x_{tia}\Big)
\end{equation}

Next, we calculate the downloading cost related to content recommendation.
Denote by $p_{ija}$ the probability of accepting content $j\in\mathcal{R}_i$ with AoI $a$ instead of content $i$. This probability depends both on the correlation between the two contents as well as the AoI of content $j$. The value of $p_{ija}$ can be calculated based on historical statistics~\cite{2018_sermpezis}, item-item recommendation~\cite{2003_linden}, and/or collaborative filtering techniques~\cite{2009_xio}. 
Denote by $\mathpzc{c}$ a generic candidate set of contents for recommendation. Because of AoI, each element of $\mathpzc{c}$ is a tuple of a recommended content and its AoI. We refer to $\mathpzc{c}$ as the recommendation set. Denote by $\mathcal{C}_{ti}$ the set of all such recommendation sets for content $i\in\mathcal{I}$ in time slot $t$.  Denote by $ v_{ti\mathpzc{c} }$ a binary optimization variable that takes value one if and only if (some content) in recommendation set $\mathpzc{c}\in$   $\mathcal{C}_{ti}$ is accepted instead of content $i$ in time slot $t$. The probability of not accepting any of the contents in $\mathpzc{c}$ is $\tilde{P}_{i\mathpzc{c}}=\prod_{(j,a) \in \mathpzc{c} } (1-p_{ija})$. Thus, the probability of accepting at least one of them is $1-\tilde{P}_{ic}$, and hence the expected cost\footnote{In this cost, $1-\tilde{P}_{i\mathpzc{c}}$ is multiplied by the size of initially requested content i.e., $s_i$. The reason is that we consider recommending only contents with similar size to $s_i$.} is:

\begin{equation}\label{delta2}
\Delta_{recom}=\sum_{t \in \mathcal{T}}\sum_{i \in \mathcal{I}}\sum_{\mathpzc{c}\in \mathcal{C}_{ti}}\left(c_b(1-\tilde{P}_{i\mathpzc{c}})+c_s\tilde{P}_{i\mathpzc{c}}\right)s_ih_{ti}v_{ti\mathpzc{c}}
\end{equation}

Finally, the total cost of system is the sum of $\Delta_{download}$ and $\Delta_{recom}$.

\subsection{Problem Formulation}\label{subsec:problem_formulation}
COPRA can be formulated using integer-linear programming (ILP), as shown in \eqref{ilp}. In \eqref{ilp}, we use  $y_{ti}$ as an auxiliary binary 
variable that equals one if and only if content item~$i$ is cached in time slot~$t$.
Constraints~\eqref{const:AoI1} state that if content $i$ is cached in time slot $t$, then 
it should exactly take one of the possible AoIs $a$ in $\mathcal{A}_{ti}$. Constraints~\eqref{const:AoI2}~and~\eqref{const:AoI3} 
together ensure that content $i$ in time slot $t$ has AoI $a$ (i.e., $x_{ita}=1$) if and only if three
conditions hold: Item $i$ is in the cache ($y_{ti} = 1$), it has AoI $a-1$
in time slot $t-1$ ($x_{i(t-1)(a-1)}=1$), and it is not refreshed
again in slot $t$ ($x_{it0} = 0$). Constraints~\eqref{const_ilp:vy} indicate that either content 
$i$ is cached in time slot $t$ or some set $\mathpzc{c} \in \mathcal{C}_{ti}$ is recommended. Constraints~\eqref{const_ilp:vx}
ensure that the contents in  recommendation  set $\mathpzc{c}$ are indeed cached. Constraints~\eqref{const_ilp:backhaulSize}~and~\eqref{const_ilp:cacheSize}
formulate the cache and backhaul capacities. Finally, Constraints~\eqref{const_ilp:domain_y}-\eqref{const_ilp:domain_v}
state the variable domain.
\begin{figure}[!ht]
\begin{subequations}\label{ilp}
\begin{alignat}{2}
\text{ILP}:&\min\limits_{\bm{y},\bm{x},\bm{v}}\quad
\Delta_{download}+\Delta_{recom} \label{ilpobj}\\
\text{s.t}. \quad
& \sum_{a\in\mathcal{A}_{ti}}x_{tia}=y_{ti},t \in \mathcal{T}, i \in \mathcal{I}\label{const_ilp:AoI1}\\
&x_{tia}\ge y_{ti}+x_{(t-1)i(a-1)}-x_{ti0}-1,\nonumber\\
&~~~~~~~~~~~~~~~~~~~t \in\mathcal{T}\setminus\{1\}, i \in \mathcal{I},a\in\mathcal{A}_{ti}\setminus\{0\}\label{const_ilp:AoI2}\\
&x_{tia}\le x_{(t-1)i(a-1)},\nonumber\\
&~~~~~~~~~~~~~~~~~~~t\in\mathcal{T}\setminus \{1\},i \in\mathcal{I},a\in\mathcal{A}_{ti}\setminus\{0\}\label{const_ilp:AoI3}\\
& \sum_{\mathpzc{c}  \in \mathcal{C}_{ti}} v_{ti\mathpzc{c} }+y_{ti}=1,t\in \mathcal{T},i \in \mathcal{I}\label{const_ilp:vy}\\
& \sum_{\mathpzc{c}  \in \mathcal{C}_{ti}: (j,a) \in \mathpzc{c} } v_{ti\mathpzc{c} }\le x_{tja},t\in \mathcal{T},i \in\mathcal{I},j \in\mathcal{R}_i \label{const_ilp:vx}\\
& \sum_{i\in \mathcal{I}}s_iy_{ti} \leq S,t\in \mathcal{T}\label{const_ilp:cacheSize}\\
& \sum_{i\in \mathcal{I}}s_ix_{ti0} \leq L,t\in \mathcal{T}\label{const_ilp:backhaulSize}\\
& y_{ti}\in\{0,1\},t \in \mathcal{T},  i \in \mathcal{I}\label{const_ilp:domain_y}\\
& x_{tia}\in\{0,1\}, t\in \mathcal{T}, i\in \mathcal{I}, a \in \mathcal{A}_{ti} \label{const_ilp:domain_x}\\
& v_{ti\mathpzc{c} }\in\{0,1\},  t\in \mathcal{T}, i \in \mathcal{I}, \mathpzc{c}  \in \mathcal{C}_{ti}\label{const_ilp:domain_v}
\end{alignat}
\end{subequations}
\end{figure}

As the number of recommendations set are exponentially many, the ILP is exponential in size. However, the ILP is of interest for solving small-scale problem instances for gauging the performance of other suboptimal solutions.

\section{Complexity Analysis}
In this section, we rigorously prove the NP-hardness of COPRA based on a reduction from the 3-SAT.
Next, we show the tractability of the problem for a single time slot when the contents are
partitioned into subcategories with uniform probabilities. 
\begin{theorem}
COPRA is NP-hard.
\end{theorem}
\begin{proof}
We adopt a polynomial-time reduction from the 3-SAT problem that is NP-complete\cite{garey1979computers}.
Consider any 3-SAT instance with $k$ clauses and $n$ Boolean variables $u_1,u_2,...,u_n$. A variable
or its negation is called a literal. Denote by $\hat{u}_i$ the negation of $u_i$, $i = 1,2,...,n$.
Each clause consists of a disjunction of exactly three different literals, for example, $\hat{u}_1 \lor u_5 \lor u_7$.
The task is to determine if there is an assignment of true/false values to the variables, such that all clauses are satisfied (i.e., at least one literal has value
true in every clause).

We construct a reduction from 3-SAT as follows. Each literal or clause represents a content, referred to as literal and clause contents, respectively. Moreover, $n$ auxiliary contents are defined, one for each pair of variable and its negation. Hence, there are in total $3n+k$ contents, and  $\mathcal{I}=\{1,2,...,3n+k\}$.
All contents have unit size, i.e., $s_{i}=1$ for $i\in\mathcal{I}$.
Each variable, its negation, and the corresponding auxiliary content are related mutually with acceptance probability of $1$. Thus, the requests made for any of them can be
fully satisfied by any of the other two contents.
The number of time slots is one, i.e., $\mathcal{T}=\{1\}$, and the size of cache is $n$, i.e., $S=n$.
The number of requests for each clause  and literal content is $1$, i.e., $h_{1i}=1$ if $i$ is a literal or a clause content.
Each clause content is related to the corresponding three literal contents with acceptance probability of $1$. Hence for a request made for a clause content, the
system can recommend the three literals if some or all of them are cached.
There are $n+k+1$ requests for
each auxiliary content, i.e., $h_{1i}=n+k+1$ if $i$ is an auxiliary content. 
No relation is present between contents other than those specified above.
Note that the acceptance probability is symmetric between any two related contents.

Parameters $c_b$ and $c_s$ are set to $1$ and $2$, respectively.
We now show there is an optimal solution such that the cache stores exactly either a variable or its negation. 
Suppose an auxiliary and/or a clause content is cached. In the former case, swapping this auxiliary content with a non-cached literal content of the corresponding pair will not increase but possibly improve the cost. Because, by swapping, more clause contents
may also be satisfied from the cache.
Now, suppose a clause content is cached. Then, at least one auxiliary content must be served using the server with cost $2(n+k+1)$.
 For the other contents, the best possible outcome is $(n-1)(n+k+1) + 2n + k+1$. Hence, the total cost is $\Delta_1=2(n+k+1)+(n-1)(n+k+1) + 2n + k+1$.
The cost when exactly one literal of each literal pair is cached is no more than $\Delta_2=3n+n(n+k+1)+2k$, assuming all clause contents are served using the server. 
 It can be verified easily that $\Delta_1>\Delta_2$.
 Therefore, at an optimum, the cache stores exactly either a variable or its negation. Thus the optimal total cost for the literal and auxiliary contents is known. 

Clearly the construction above is polynomial in size.
If there is no solution to the 3-SAT, then at least one clause content need to be downloaded from server with 
cost $c_s=2$, and each of the other clause contents has at least the cost of $c_b=1$. 
Thus, the total cost is at least $\delta_1=k+1$.
If there is a solution to 3-SAT, then the cost for all clause contents is at most $\delta_2=k$. 
As can be seen $\delta_1>\delta_2$. Thus, whether or not there exists a caching
 strategy with a total cost of no more than $\delta_2$ gives the correct answer to 3-SAT.
Therefore, the recognition versions of COPRA is NP-complete and its optimization version is NP-hard.
\end{proof}

In practice, the content items may naturally fall into different subcategories based on the type of the content, e.g., video contents can be catagorized based on if it is science fiction,
drama, or comedy, etc.,~\cite{2020_fu,2021_fu}.
If all items in a subcategory are related with the same acceptance probability, then we show
the optimal solution of the problem with uniform size and one time slot can be computed in polynomial time via DP. Note that
the probability from a subcategory to another may still differ. We refer to this special case as COPRA-CAT.

\begin{theorem}
COPRA-CAT can be solved in polynomial time.
\end{theorem}
\begin{proof}
We compute a matrix, called cost matrix and
denote it by $\bm{g}$, in which entry $g(k, i)$ represents the 
total cost by caching $i$ content items of category $k$.
This value is computed  simply from the first $i$ contents with the highest requests.
Below, a recursive function is introduced to derive the
optimal caching solution over all categories. We define a second
matrix, called the optimal cost matrix, and denote it by $\bm{w}$,
in which $w(k,s')$ represents the cost of the optimal solution
by considering the first $k$ categories using a cache size of $s'$, 
$s' = 0, 1,...,S$. The value of $w(q, s')$ is computed by the
following recursion:
\begin{equation}\label{uniform_1slot}
\begin{aligned}
 w(k, s')= \underset{r=0,1,...,s'}{\min}\{g(k,r)+w(k-1,s'-r)\}
\end{aligned}
\end{equation}
Using Equation~\eqref{uniform_1slot}, the optimal
cost for the first $k$ categories is computed given the optimal cost of the first
$k-1$ categories. For the overall solution, the optimal cost can be computed
using the above recursion for cache size $S$ and $K$ categories.
We prove it by induction. First, when $k= 1$, i.e., we have only one category,
We have $w(1,s') = \min_rg(1,r)$ for all $r=0,1,...,s'$. Obviously $r^*=s'$, that is, to allocate the whole 
capacity to this category. Now, assume $w(l,s')$ is optimal for some $l$. We prove that $w(l+1,s')$ is
optimal. According to the recursive function:

\begin{equation}
\begin{aligned}
 w(l+1, s')= \underset{r}{\min}\{g(l+1,r)+w(l,s'-r)\}
\end{aligned}
\end{equation}
The possible values for $r=0,1,...,s'$, and
for each of the possible values of $r$, $w(l, s'-r)$ is optimal.
This together gives the conclusion that the minimum will be
obtained indeed by the $\min$ operation. Thus, $w(l+1, s')$ is
optimal.

Finally, we show that $w(K, S)$ can be computed in polynomial time.
The complexity of computing $g$ is of $O(KI)$.  By the above, the computational
complexity of $\bm{w}$ is of $O(KS^2)$ where $S$ is up to the number of contents.
\end{proof}

\section{Greedy Algorithm}
A commonly considered strategy for fast but suboptimal solution is a greedy approach (GA) that builds up a solution incrementally.
GA tries to minimize the total cost of each time slot by considering the items one by one.
The algorithm is shown in Algorithm~$\ref{alg:GA}$.

For each time slot and each item, GA calculates an overall score based on the number of requests, the relations to other contents, and the size of the content, see
Line~$\ref{alg_ga:score}$.
Then, GA treats items based on their scores in descending order. For a content under processing, it is downloaded from server  to the cache if there is enough cache and backhaul capacities,
see Lines~$\ref{alg_ga:update1}$-$\ref{alg_ga:update2}$.
Otherwise, GA checks if the content is cached in the previous time slot, and if there is enough cache capacity to store the content, see Lines~$\ref{alg_ga:cache1}$-$\ref{alg_ga:cache2}$.
When all contents are processed, GA finds recommendation sets for the non-cached items. For each non-cached item, GA looks at the cached and related items, and pick the ones of highest acceptance probabilities, see Lines~$\ref{alg_ga:recom1}$-$\ref{alg_ga:recom2}$.
GA is simple but it turns out the performance is not satisfactory, and therefore there is a need of developing a better algorithm.

\begin{algorithm}[ht!]
{\bf Input}: $\mathcal{T}$, $S$, $L$, $\bm{p}$, $\bm{h}$, $A_i$, $s_i, i \in \mathcal{I}$\\
{\bf Output}: $\bm{y}^*$, $\bm{x}^*$, $\bm{v}^*$\\
\caption{Greedy Algorithm}\label{alg:GA}
\begin{algorithmic}[1]
\STATE $\bm{y}\leftarrow\{0,t\in\mathcal{T}\cup\{0\},i\in\mathcal{I}\}$
\STATE $\bm{x}\leftarrow\{0,t\in\mathcal{T},i\in\mathcal{I}, a\in\mathcal{A}_{ti}\}$
\STATE $\mathcal{C}'_{ti}\leftarrow \emptyset, t\in\mathcal{T}, i\in\mathcal{I}$ 
\STATE $\bm{v}\leftarrow\{0,t\in\mathcal{T},i\in\mathcal{I}, \mathpzc{c} \in\mathcal{C}'_{ti}\}$
\FOR{$t\in\mathcal{T}$}
\STATE $\mathcal{I}'\leftarrow \mathcal{I}$, $L'\leftarrow L$, $S'\leftarrow S$
\STATE $\theta_{ti}\leftarrow \nicefrac{(\sum_{j\in\mathcal{I}}p_{ji0}h_{tj})}{s_i}$,~{$i\in\mathcal{I}$}\label{alg_ga:score}
\WHILE{$\mathcal{I}'\neq \emptyset$~and~$S'>0$}
\STATE $i^*\leftarrow \arg\max_{i\in \mathcal{I}'}\theta_{ti}$
\STATE $\mathcal{I}'\leftarrow \mathcal{I}'\setminus\{i^*\}$
\IF {($s_{i^*}\le L'$ and $s_{i^*}\le S'$)}\label{alg_ga:update1}
\STATE $L'\leftarrow L'-s_{i^*}$, $S'\leftarrow S'-s_{i^*}$
\STATE $y_{ti^*}\leftarrow 1$, $x_{ti^*0}\leftarrow 1$\label{alg_ga:update2}
\ELSIF{($y_{(t-1)i^*}=1$ and $s_{i^*}\le S'$ and $x_{ti^*A_i^*}\neq1$)}\label{alg_ga:cache1}
\STATE $S'\leftarrow S'-s_{i^*}$
\STATE $a^*\leftarrow \arg\max_{a}x_{(t-1)i^*a}$
\STATE $y_{ti^*}\leftarrow1$, $x_{ti^*(a^*+1)}\leftarrow1$\label{alg_ga:cache2}
\ENDIF
\ENDWHILE
\FOR{$i\in\mathcal{I}'$}\label{alg_ga:recom1}
\STATE $\mathpzc{c} \leftarrow\{(j,a):x_{tja}=1,  j\in \mathcal{R}_i, a\in\mathcal{A}_{tj}\}$
\STATE $\mathcal{C}'_{ti} \leftarrow\{$the first $N$ elements in $\mathpzc{c}$ with the highest probabilities \\ with respect to $i\}$
\STATE $ v_{ti\mathpzc{c} }\leftarrow 1$
\ENDFOR
\ENDFOR \label{alg_ga:recom2}
\end{algorithmic}
\end{algorithm}

\section{Algorithm Design}
We propose an algorithm by applying Lagrangian decomposition (LD) to ILP~\eqref{ilp}.
In LD, some variables are duplicated, with equalities constraints requiring that the duplicates are equal to the original variables. 
Next, these constraints are relaxed  using Lagrangian relaxation and some method (often a subgradient method~\cite{1977_goffin,1981_bazara}) is applied to solve resulting Lagrangian dual. 


\subsection{Lagrangian Decomposition}
In our LD-based algorithm (LDA), we duplicate the $\bm{x}$ variables. Specifically, we replace $\bm{x}$ variables in AoI constraints~$\eqref{const_ilp:AoI1}$-$\eqref{const_ilp:AoI3}$ by $\bm{x}'$ and add a set of constraints  requiring $\bm{x}=\bm{x}'$. Next, we relax constraints $\bm{x}=\bm{x}'$ with multipliers $\bm{\lambda}$, and the resulting Lagrangian relaxation is given~in~\eqref{ld}. Note that $\Delta'_{download}$ is the same as $\Delta_{download}$ but the $x$ variables are replaced by $x'$.

\begin{figure}[!ht]
\begin{subequations}\label{ld}
\begin{alignat}{2}
&\min\limits_{\bm{y},\bm{x},\bm{x'},\bm{v}\in\{0,1\}}\Delta'_{download}+\Delta_{recom}+\nonumber \\
&~~~~~~~~~~~~~~~~~~~~~~~~~~~~\sum_{t\in\mathcal{T}}\sum_{i \in\mathcal{I}}\sum_{a\in\mathcal{A}_{ti}}\lambda_{tia}(x'_{tia}-x_{tia}) \label{ld_obj}\\
&\text{s.t}. \quad
\sum_{a\in\mathcal{A}_{ti}}^{t-1}x'_{tia}=y_{ti},t \in \mathcal{T}, i \in \mathcal{I}\label{const_ld:AoI1}\\
&x'_{tia}\ge y_{ti}+x'_{(t-1)i(a-1)}-x'_{ti0}-1,\nonumber\\
&~~~~~~~~~~~~~~~~~~~~~~~~t \in\mathcal{T}\setminus\{1\}, i \in \mathcal{I},a\in\mathcal{A}_{ti}\setminus\{0\}\label{const_ld:AoI2}\\
&x'_{tia}\le x'_{(t-1)i(a-1)},t \in\mathcal{T}\setminus \{1\},i \in\mathcal{I},a\in\mathcal{A}_{ti}\setminus\{0\}\label{const_ld:AoI3}\\
& \eqref{const_ilp:vy}-\eqref{const_ilp:backhaulSize}\nonumber
\end{alignat}
\end{subequations}
\end{figure}

As can be seen ILP~\eqref{ld} is decomposed into to two subproblems, one consists of all terms containing $\bm{x}'$, and the other all terms with $\bm{x}$. 
Below, we formally state each of them.
 
 \subsection{Subproblem~One}
 Subproblem~$1$, hereinafter referred to as SP$_1$, is shown in~\eqref{sp1}. The SP$_1$ consists of all terms having $\bm{x}'$, namely  the downloading cost and Lagrangian multiplier terms as the objective function, and constraints related to AoI. 
 
 \begin{figure*}
\centering
\includegraphics[scale=0.55]{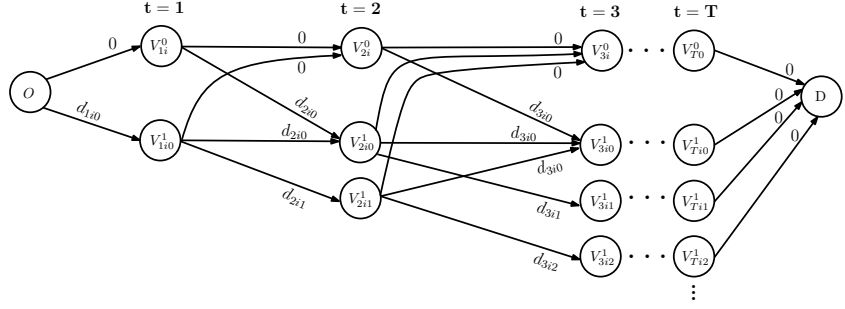}
\begin{center}
\vspace{-3mm}
\caption{Graph of the shortest path for SP$_{1}^{(i)}$ corresponding to content item $i$.}
\label{fig:sp1t}
\end{center}
\end{figure*}

 \begin{figure}[!ht]
\begin{subequations}\label{sp1}
\begin{alignat}{2}
\text{SP}_1:&\min\limits_{\bm{y},\bm{x'}\in\{0,1\}}
\Delta'_{download}+
\sum_{t\in\mathcal{T}}\sum_{i \in\mathcal{I}}\sum_{a\in\mathcal{A}_{ti}}\lambda_{tia}x'_{tia} \label{obj:sp1}\\
&\text{s.t}. \quad
\sum_{a\in\mathcal{A}_{ti}}x'_{tia}=y_{ti},t \in \mathcal{T}, i \in \mathcal{I}\label{const:AoI1}\\
&x'_{tia}\ge y_{ti}+x'_{(t-1)i(a-1)}-x'_{ti0}-1,\nonumber\\
&~~~~~~~~~~~~~~~~t \in\mathcal{T}\setminus\{1\}, i \in \mathcal{I},a\in\mathcal{A}_{ti}\setminus\{0\}\label{const:AoI2}\\
&x'_{tia}\le x'_{(t-1)i(a-1)},\nonumber\\
&~~~~~~~~~~~~~~~t\in\mathcal{T}\setminus \{1\},i \in\mathcal{I},a\in\mathcal{A}_{ti}\setminus\{0\}\label{const:AoI3}
\end{alignat}
\end{subequations}
\end{figure} 
We exploit the structure of SP$_1$ as follows.
 First, as there is  no constraint bundling the content items together, SP$_1$ decomposes by content, leading to $I$ smaller problems.
 The optimization problem corresponding to content $i\in\mathcal{I}$ is denoted by SP$_1^{(i)}$  and consists of the terms of SP$_1$ for content $i$.
 Second, we show that SP$_1^{(i)}$, $i \in\mathcal{I}$, can be solved as a shortest path problem.

 \begin{theorem}
 SP$_1^{(i)}$, $i \in\mathcal{I}$, can be solved in polynomial time as a shortest path problem.
\end{theorem}
\begin{proof}
Consider content $i\in \mathcal{I}$. We construct an acyclic
directed graph such that finding the shortest path from the origin
to distention is equivalent to solving SP$_1^{(i)}$. The graph is shown in Figure~\ref{fig:sp1t}. The objective function of SP$_1^{(i)}$ is:

\begin{equation}\label{obj_sp1i_rewrite}
\begin{aligned}
&\sum_{t \in \mathcal{T}}\Big((c_s-c_b)s_ix'_{ti0}+\sum_{a\in\mathcal{A}_{ti}}c_bs_ih_{ti}f_{tia}x'_{tia}\Big)+\sum_{t\in\mathcal{T}}\sum_{a\in\mathcal{A}_{ti}}\lambda_{tia}x'_{tia}\\
&=\sum_{t \in \mathcal{T}}(c_s-c_b)s_ix'_{ti0}+\sum_{t \in \mathcal{T}}\sum_{a\in\mathcal{A}_{ti}}(c_bs_ih_{ti}f_{tia}+\lambda_{tia})x'_{tia}
\end{aligned}
\end{equation}

The graph is constructed as follows. Nodes $O$ and $D$ are
used to represent the origin and destination respectively. For time slot $t$, 
there are $1+\min\{A_i,t-1\}$ vertically aligned nodes. A path passing node $V_{ti}^0$ and $V_{tia}^1$ corresponds to 
the following two scenarios, respectively: 

\begin{enumerate}
    \item The content is not in the cache.
\item The content is in the cache and has AoI $a$, $a\in\mathcal{A}_{ti}$.

\end{enumerate}

 For each node $V_{ti}^0$ there are two
outgoing arcs, one to $V_{(t+1)i}^0$ which corresponds to that the content is
not stored in the next time slot and the arc hence has weight zero, and the other
to $V_{(t+1)i0}^1$ which has weight $d_{(t+1)i0}=(c_s-c_b)s_i+c_bs_ih_{(t+1)i}f_{(t+1)i0}+\lambda_{(t+1)i0}$ corresponding to the case that the content
is downloaded to the cache in the next time slot and has AoI
zero. For each node $V_{tia}^1$ there three outgoing arcs to $V_{(t+1)i}^0$,
$V_{(t+1)i0}^1$, and $V_{(t+1)i(a+1)}^0$, respectively. 
A path passing through the first, second, and the third arcs corresponds to the following three scenarios, respectively:

\begin{enumerate}
    \item Content is deleted for the next
time slot with arc weight zero.
\item The content is re-downloaded
from the cache and has AoI zero with weight $d_{(t+1)i0}$.
\item The content is kept and its AoI increases
with one time unit and has weight $d_{(t+1)i(a+1)}=c_bs_ih_{(t+1)i}f_{(t+1)i(a+1)}+\lambda_{(t+1)i(a+1)}$.
\end{enumerate}
Finally, there are $T$ arcs from $V_{Ti0}^0$ and $V_{Tia}^1$ to $D$,
each with weight zero.

Given any solution of SP$_{i}^{(i)}$, by construction of the graph,
the solution directly maps to a path from the origin to the
destination with the same objective function. Conversely, given
a path we construct an ILP solution. For time slot $t$, if the path contain
node $V_{ti}^0$ and $V_{tia}^1$, we set $y_{ti}=0$. If the path passes through node $V_{tia}^1$, we
set $y_{ti} = 1$ and $x'_{tia}=1$. The resulting ILP solution has the
same objective function value as the path length in
terms of the arc weights. Hence the conclusion.
\end{proof}


\subsection{Subproblem~Two}
Subproblem~$2$, hereinafter referred to as SP$_2$, consists of all those terms of \eqref{ld} containing $\bm{x}$. 
SP$_2$ decomposes by time slot,
leading to $T$ smaller problems. Denote by SP$_2^{(t)}$ the problem corresponding to
time slot $t$, shown in \eqref{sp2t}. 

The number of $v$ variables in SP$_2^{(t)}$ is exponentially many, as there are exponential number of recommendation sets. Hence, having
all $v$ variables in the ILP is impractical. To deal with this issue, we apply column generation to the $\bm{v}$ variables in the LP relaxation of~\eqref{sp2t}, to generate only the promising recommendation sets. Column generation is a powerful method to
obtain the global optimum of some structured linear programs with exponential number of variables. In a column generation algorithm, the most promising variables are generated in a iterative process
by solving alternatively a  master problem (MP) and a pricing problem (PP). Each time PP is solved, a new variable that possibly
improves the objective function is generated. The benefit of column generation is to exploit the fact that at optimum only a few variables are positive.
Below we  define the MP and PP for solving SP$_{2}^{(t)}$. In the following,  to ease the presentation, we consider a generic time
slot and drop the index $t$.

\begin{figure}[h]
\begin{subequations}\label{sp2t}
\begin{alignat}{2}
\text{SP}_2^{(t)}:&\min\sum_{i \in \mathcal{I}}\sum_{\mathpzc{c} \in \mathcal{C}_{i}}\left(c_b(1-\tilde{P_{ic}})+c_s\tilde{P_{ic}}\right)s_ih_{ti} v_{ti\mathpzc{c} }\\
&~~~~~~~~~~~~~~~~~~~~~~~~~~~~~~~~~~~~~~~~-\sum_{i \in\mathcal{I}}\sum_{a\in\mathcal{A}_{ti}}\lambda_{tia}x_{tia}\nonumber \\
\label{obj:sp2t}\\
\text{s.t}. \quad
& \sum_{a\in\mathcal{A}_{ti}}x_{tia}=y_{ti},i \in \mathcal{I}\label{const_sp2t:AoI1}\\
& \sum_{\mathpzc{c} \in \mathcal{C}_i} v_{ti\mathpzc{c} }+y_{ti}=1,i \in \mathcal{I}\label{const_sp2t:vy}\\
& \sum_{\mathpzc{c} \in \mathcal{C}_i: (j,a) \in \mathpzc{c} } v_{ti\mathpzc{c} }\le x_{tja},,i \in\mathcal{I},j \in\mathcal{R}_i \label{const_sp2t:vx}\\
& \sum_{i\in \mathcal{I}}s_iy_{ti} \leq C\label{const_sp2t:cacheSize}\\
& \sum_{i\in \mathcal{I}}s_ix_{ti0} \leq L\label{const_sp2t:backhaulSize}\\
& y_{ti}\in\{0,1\},t \in \mathcal{T},  i \in \mathcal{I}\label{const_sp2t:domain_y}\\
& x_{tia}\in\{0,1\}, t\in \mathcal{T}, i\in \mathcal{I}, a \in \mathcal{A}_{ti} \label{const_sp2t:domain_x}\\
&  v_{ti\mathpzc{c} }\in\{0,1\}, t\in\mathcal{T},i \in \mathcal{I}, \mathpzc{c} \in\mathcal{C}_{ti}\label{const_sp2t:domain_v}
\end{alignat}
\end{subequations}
\end{figure}

\subsubsection{MP and RMP}
MP is the continuous version~of~\eqref{sp2t}. Restricted MP~(RMP) is the~MP but with a small
subset~$\mathcal{C}'_i \subseteq \mathcal{C}_i$ for any content  $i\in \mathcal{I}$. Denote by~$C'_i$ the cardinality of~$\mathcal{C}'_i$.

\subsubsection{Pricing problem}
The PP uses the dual information to generate new variables/columns. Denote by $\bm{v}^*$ the optimal solution of RMP. 
After obtaining $\bm{v}^*$, we need to check whether $\bm{v}^*$ is the
global optimum of RMP. This can be determined by finding
a column with the minimum reduced cost for each content
$i\in\mathcal{I}$. This means the PP decomposes to $I$ smaller problems, one corresponding to each content $i$.
 If all these minimum reduced cost values are nonnegative, the current solution is optimal. 
Otherwise, we add the columns with negative
reduced costs to their recommendation sets.

Consider content $i\in\mathcal{I}$. Denote by $\pi^*_i$ and $\bm{\beta}^*_i=\{\beta_{ija}|j\in \mathcal{R}_i,a\in\mathcal{A}_i\}$ the optimal dual values of the counterpart constraints of \eqref{const_sp2t:vy} and \eqref{const_sp2t:vx} in the RMP, respectively. 
Hence, the reduced cost of the $v$-variable of content $i$ and recommendation set $\mathpzc{c}$ is:
\begin{equation}\label{1objnonlin_ppi}
\begin{aligned}
&\left(c_b(1-\tilde{P_{i\mathpzc{c} }})+c_s\tilde{P_{i\mathpzc{c} }}\right)s_ih_{i}-\pi^*_{i}+\sum_{j\in\mathcal{R}_i}\sum_{a\in\mathcal{A}_{j}}\beta^*_{ija}=\\
&\left(s_ih_i(c_s-c_b)\tilde{P_{i\mathpzc{c} }}\right)+c_bs_ih_i-\pi^*_{i}+\sum_{j\in\mathcal{R}_i}\sum_{a\in\mathcal{A}_j}\beta^*_{ija}
\end{aligned}
\end{equation}
in which $\tilde{P}_{i\mathpzc{c}}=\prod_{(j,a) \in \mathpzc{c} } (1-p_{ija})$. 
This reduced cost is nonlinear due to the term  $\tilde{P}_{i\mathpzc{c}}$. But, we can linearize it using logarithm. Let 

\begin{equation}\label{2objnonlin_ppi}
\begin{aligned}
p'&=\log\left(h_is_i(c_s-c_b)\prod_{(j,a) \in \mathpzc{c} } (1-p_{ija})\right)\\
&=\log\left(h_is_i(c_s-c_b)\right) +\sum_{(j,a) \in \mathpzc{c} }\log(1-p_{ija})
\end{aligned}
\end{equation}

Now, the reduced cost can expressed as:
\begin{equation}\label{3objnonlin_ppi}
10^{p'}+c_bs_ih_i-\pi^*_{i}+\sum_{j\in\mathcal{R}_i}\sum_{a\in\mathcal{A}_j}\beta^*_{ija}
\end{equation}
where $p'=\log\left(h_is_i(c_s-c_b)\right) +\sum_{(j,a) \in \mathpzc{c} }\log(1-p_{ija})$. As $p_{ija}\in (0,1)$, $\sum_{(j,a) \in \mathpzc{c} }\log(1-p_{ija})$ is zero or a negative value. Thus, the minimum and maximum values that $p'$ can take are $p'_{min}=\log\left(h_is_i(c_s-c_b)\right) +\sum_{(j,a) \in \mathpzc{c} }\log(1-p_{ija})$. 
 and $p'_{max}=\log\left(h_i(c_s-c_b)\right)$, respectively. Hence $p'\in[p_{min},p_{max}]$.
The above expression is for a given $v$-variable. In the following, we define PP, that is an auxiliary problem, of which the optimum will tell us the not-yet-present variable with minimum reduced cost.

Denote by PP$^{(i)}$ the PP corresponding to content $i$. Let $z_{ja}$ be a binary optimization variable that takes value
one if and only if content $j$ with AoI $a$ is in the set to be generated. Then PP$^{(i)}$ can be expressed~as~\eqref{ppi}. Note that the terms $c_bs_ih_i$ and $\pi^*_{i}$ are constants here, 
and hence can be dropped in the optimization process.
Constraints~\eqref{const_ppi:AoI} ensure that for each content in the recommendation set, exactly one AoI value is selected.
Constraint \eqref{const_ppi:rmax} states that the total number of contents in the recommendation set can not exceed the given upper bound. In the following, we show that PP$^{(i)}$ can be solved via DP.
\begin{figure}[H]
\begin{subequations}\label{ppi}
\begin{alignat}{2}
\text{PP}^{(i)}:&\min 10^{p'}+\sum_{j\in\mathcal{R}_i}\sum_{a\in\mathcal{A}_j}\beta^*_{ija}z_{ja}\label{ppi:obj}\\
\text{s.t}. \quad
&p'=\log\left(h_is_i(c_s-c_b)\right) +\sum_{j\in \mathcal{R}_i}\sum_{a\in \mathcal{A}_j}\log(1-p_{ija})z_{ja}\label{const_ppi:eq}\\
& \sum_{a\in\mathcal{A}_j}z_{ja}\le1,j\in\mathcal{R}_i\label{const_ppi:AoI}\\
& \sum_{j \in \mathcal{R}_i}\sum_{a\in \mathcal{A}_j}z_{ja}\le N\label{const_ppi:rmax}\\
&z_{ja}\in\{0,1\},j \in \mathcal{R}_i,a\in\mathcal{A}_{j}\\
&p'\in[p_{min},p_{max}]
\end{alignat}
\end{subequations}
\end{figure}
\vspace{-7mm}
\begin{theorem}
PP$^{(i)}$ can be solved to any desired accuracy via DP.
\end{theorem}
\begin{proof}
We first perform two prepossessing steps, and then apply DP to the resulting problem.
First, as the objective function is minimization and $p'$ is a continuous variable, constraint~\eqref{const_ppi:eq} can be stated equivalently as a greater-than-or-equal constraint.
We re-express the constraint as:
\begin{equation}\label{ineq:11c}
\begin{aligned}
\sum_{j\in \mathcal{R}_i}\sum_{a\in \mathcal{A}_j}-\log(1-p_{ija})z_{ja}\ge \log\left(h_is_i(c_s-c_b)\right)-p'
\end{aligned}
\end{equation}
Since $p'\in [p_{min},p_{max}]$, the minimum and maximum values that the right-hand-side of the constraint can take are zero and $\sum_{j\in \mathcal{R}_i}\sum_{a\in \mathcal{A}_j}-\log(1-p_{ija})$, respectively.

Second, the problem  can be solved to any desired accuracy (though not exactly the optimum), by quantizing the interval of $p'$ into $W$  steps; this corresponds to multiplying the coefficients with some number $M$ and rounding.
Let $W=M\sum_{j\in \mathcal{R}_i}\sum_{a\in \mathcal{A}_j}-\log(1-p_{ija})$ denote the maximum value of the right-hand-side of \eqref{ineq:11c} after multiplying it by $M$. 
Similarly let $q_{ja}=-M\log(1-p_{ija})$ for $j\in\mathcal{R}_i$, $a\in\mathcal{A}_j$.
Note that the minimum value 
that the right-hand-side can take is still zero. Hence, $p'\in[0,W]$. 

After these two steps, \eqref{ppi} can be re-expressed as \eqref{ppi_re}.
Formulation \eqref{ppi_re} resembles  an inversed multiple-choice knapsack problem with an upper bound \eqref{const_ppi_re:rmax} on the number of items. The difference is that we have $p'$ as one additional variable with a term in the objective function. The selection of $p'$ affects the right-hand-side, corresponding to changing the knapsack capacity. Knapsack problem is solved via DP efficiently. The interesting point is that DP provides not only the optimal function value of $z$ with the given capacity, but also those for all intermediate values starting from zero, implying that one computation is enough to examine the effect of all $p'$ values. Then the optimum can be obtained by post processing considering the function term with~$p'$.
\end{proof}
\vspace{-7mm}
\begin{figure}[H]
\begin{subequations}\label{ppi_re}
\begin{alignat}{2}
\centering
&\min 10^{p'}+\sum_{j\in\mathcal{R}_i}\sum_{a\in\mathcal{A}_j}\beta^*_{ija}z_{ja}\label{ppi_re:obj}\\
\text{s.t}. \quad
&\sum_{j\in \mathcal{R}_i}\sum_{a\in \mathcal{A}_j}q_{ja}z_{ja}\ge W \label{const_ppi_re:capa}\\
& \sum_{a\in\mathcal{A}_j}z_{ja}\le1,j\in\mathcal{R}_i\label{const_ppi_re:AoI}\\
& \sum_{j \in \mathcal{R}_i}\sum_{a\in \mathcal{A}_j}z_{ja}\le N\label{const_ppi_re:rmax}\\
&z_{ja}\in\{0,1\},j \in \mathcal{R}_i,a\in\mathcal{A}_{j}\\
&p'\in[0,W]
\end{alignat}
\end{subequations}
\end{figure}
\vspace{-5mm}
The DP algorithm for solving PP$^{(i)}$ is shown in Algorithm~\ref{alg:DPA}.
Lines~\ref{ini_dp0}-\ref{ini_dp1} are the initialization steps.
Lines~\ref{while1_dp0}-\ref{while1_dp1} solves \eqref{ppi_re} with maximum capacity $W$ in which matrix $\bm{B}^*$ is the optimal cost matrix. Entry $B^*(w,j,n)$ represents the cost of the optimal solution when up to $n$ contents of the first $j$ contents can be in the recommendation set with a knapsack capacity $w\in[0,W]$. Matrix $\bm{A^*}$ is an auxiliary matrix that stores the AoI corresponding to optimum for each tuple $(w,j,n)$ where $w\in[0,W]$, $j=1,..,R_i$, and $n=1,2,...,\min\{j,N\}$.
Lines~\ref{while2_dp0}-\ref{while2_dp1} perform the post processing step. Namely for each intermediate value $p'\in [0,W]$, the corresponding objective function value is calculated and compared to the minimum value found so for, in order to find the global minimum of the problem. The complexity of this algorithm is of $O(WR_iNA_i)$. 
The column generation algorithm for solving SP$_2^{(t)}$ is shown in Algorithm~\ref{alg:CGA} in which Algorithm~\ref{alg:DPA} is used for solving PP$^{(i)}$, $i\in\mathcal{I}$.

\vspace{5mm}
\begin{algorithm}[!ht]
\caption{Column generation for SP$_2^{(t)}$}\label{alg:CGA}
\begin{algorithmic}[1]
\STATE Initialize $\mathcal{C}'_{i}$~for~$i\in \mathcal{I}$
\STATE $\text{Stop}\leftarrow\text{False}$
\WHILE{$\text{Stop}=\text{False}$}
\STATE Solve RMP~\eqref{sp2t} and obtain dual optimum values $\bm{\pi}$ and $\bm{\beta}$
\STATE $\text{Stop}\leftarrow\text{True}$
\FOR{$i\in\mathcal{I}$}
\STATE Solve PP$^{(i)}$ by Algorithm~\ref{alg:DPA}
\IF{the reduced cost$<0$}
\STATE Stop $\leftarrow$ False
\STATE Add the column to $\mathcal{C}'_i$
\ENDIF
\ENDFOR
\ENDWHILE                                                                   
\end{algorithmic}
\end{algorithm}

\begin{algorithm}[!ht]
\caption{Dynamic programming for PP$^{(i)}$}\label{alg:DPA}
\begin{algorithmic}[1]
\STATE Create matrix $B^*$ of size $(1+W)\times (1+R_i)\times (1+N)$\label{ini_dp0}
\STATE Create matrix $A^*$ of size $(1+W)\times (1+R_i)\times (1+N)$
\STATE $B^*[0,j,n]\leftarrow 0$ for~ any $j$ and $n$
\STATE $B^*[w,0,n]\leftarrow \infty$ for~any $w>0$ and any $n$
\STATE $w\leftarrow 1$,~$\text{Stop}\leftarrow\text{False}$ \label{ini_dp1}
\WHILE{$\text{Stop}=\text{False}$}\label{while1_dp0}
\FOR{$j=1,...,R_i$}
\FOR{$n=1,...,\min\{j,N\}$}
\STATE $B^*(w,j,n)\leftarrow\underset{a\in\mathcal{A}_j}\min\{\beta_{ija}+B^*(w',j-1,n-1),B^*(w,j-1,n')\}$,
$A^*(w,j,n)\leftarrow\underset{a\in\mathcal{A}_j}{\arg\min}\{\beta_{ija}+B^*(w',j-1,n-1),B^*(w,j-1,n')\}$,\\
 where \\
 $w'=\max\{0,w-q_{ja}\}$ and $n'=\min\{j-1,n\}$
\ENDFOR
\ENDFOR
\IF{$B^*(w,R_i,N)=\infty~\text{or}~w=W$}
\STATE  $\text{Stop}\leftarrow\text{True}$
\ELSE
\STATE $w\leftarrow w+1$
\ENDIF
\ENDWHILE \label{while1_dp1}
\STATE  $\text{OPT}\leftarrow\infty$,~$w\leftarrow1$,~$\text{Stop}\leftarrow\text{False}$  \label{while2_dp0}
\WHILE{$\text{Stop}=\text{False}$}
\STATE $q^*\leftarrow0$,~$w'\leftarrow w$,~$j'\leftarrow R_i$,~$n'\leftarrow N$
\IF{$B^*(w',j',n')=\infty$~or~$w>W$}
\STATE $\text{Stop}\leftarrow\text{True}$
\ELSE
\WHILE{$j'\ge1$~and~$n'\ge1$}
\IF {$B^*(w',j',n')<B^*(w',j'-1,n')$}
\STATE $a^*\leftarrow A^*(w',j',n')$
\STATE $w'\leftarrow \min\{0,w'-q_{j'a^*}\}$
\STATE $q^*\leftarrow q^*+q_{j'a^*}$
\STATE $n'\leftarrow n'-1$
\ELSE
\STATE $n'\leftarrow \min\{j'-1,N\}$
\ENDIF
\STATE $j'\leftarrow j'-1$
\ENDWHILE
\STATE $p\leftarrow \log\left(h_is_i(c_s-c_b)\right)-q^*/M$
\IF{$10^p+c_bs_ih_i-\pi^*_{i}+B^*(w,R_i,N)<\text{OPT}$ }
\STATE $\text{OPT}\leftarrow {10}^p+c_bs_ih_i-\pi^*_{i}+B^*(w,R_i,N)$
\ENDIF
\ENDIF
\STATE $w\leftarrow w+1$
\ENDWHILE \label{while2_dp1}
\end{algorithmic}
\end{algorithm}

\subsection{Attaining Integer Feasible Solutions}
The solutions of the two subproblems will likely violate some original constraints, and we present an approach to generate feasible solutions based on SP$_2$.
We take the solutions of SP$_2^{(t)}$, $t\in\mathcal{T}$, and ``repair'' them in order to construct an integer solution for COPRA. 
The reason of using SP$_2^{(t)}$, $t\in\mathcal{T}$ is that its solution contains the information of recommendation sets, 
and hence it resembles more a solution to the original problem.
However, these solutions do~not respect the AoI evolution of contents across  the time slots as each SP$_2^{(t)}$, $t\in\mathcal{T}$, is solved independently from the others. The repairing algorithm~(RA) is shown in Algorithm~$\ref{alg:RA}$, which consists of three main steps.
 In the algorithm, symbol $\leftarrow$ is used to indicate the assignment of a value.
Symbol $\leftleftarrows$ is used to indicate that an assigned value of an
optimization variable is kept fixed subsequently.

In the first step, we take the solution of SP$_2^{(t)}$, $t\in\mathcal{T}$, and perform an iterative rounding
process on the $y$-variables to obtain an integer solution. More specifically, we first fix the current $y$-variables having value one, 
followed by fixing the variable with the largest fractional value to one if there is enough capacity and zero otherwise. We then solve SP$_2^{(t)}$ again. Now, if the
solution is integer, we stop. Otherwise, this process is repeated until
an integer solution is obtained. Obviously, in the worst case, $I$ iterations are needed to obtain an integer solution. Denote by $\hat{\bm{y}}=\{\hat{y}_{ti}:t \in \mathcal{T}~\text{and}~i\in\mathcal{I}\}$ the obtained values of $y$-variables of each SP$_2^{(t)}$ for $t\in\mathcal{T}$.
This step corresponds to Lines~$\ref{step1:start}$-$\ref{step1:end}$ in Algorithm~$\ref{alg:RA}$.

In the second step, we utilize $\hat{\bm{y}}$ 
as input to the optimization problem stated in \eqref{repair}. 
Therein, the $y$-variables have the same meaning as defined earlier in Section~\ref{subsec:problem_formulation}.
Solving~\eqref{repair} provides a caching solution 
in which the AoI evolution of contents across time slots are respected. 
The objective function is maximization, to encourage setting the $\bm{y}$-variables to be as similar to $\hat{\bm{y}}$ as possible.
Here, $\epsilon$ is a
small positive number, to encourage caching contents even if $\hat{y}$ is zero.
This step corresponds to Line~$\ref{step2}$ in Algorithm~$\ref{alg:RA}$.

\begin{figure}[!ht]
\begin{subequations}\label{repair}
\begin{alignat}{2}
&\max\limits_{\bm{y},\bm{x}\in\{0,1\}}\quad
\sum_{t\in\mathcal{T}}\sum_{i \in\mathcal{I}}(\epsilon+h_{ti}s_i\hat{y}_{ti})y_{ti} \label{obj_repair}\\
&\text{s.t.}~~~~\eqref{const_ilp:AoI1}, \eqref{const_ilp:AoI2}, \eqref{const_ilp:AoI3},
\eqref{const_ilp:cacheSize},\eqref{const_ilp:backhaulSize}\nonumber
\end{alignat}
\end{subequations}
\end{figure}

After these two steps, we have  a complete caching solution over time slots. Finally, for each non-cached content item,
we choose the $N$ highest related cached contents as its recommendation set.
This step corresponds to Lines~$\ref{step3:start}$-$\ref{step3:end}$ in Algorithm~$\ref{alg:RA}$.
 We remark that formulation~\eqref{repair} is an integer program. However in practice this is solved rapidly. Moreover, the repairing operation does not need to be done in every iteration of subgradient optimization.

\begin{algorithm}[ht!]
\textbf{Input}:  SP$_{2}^{(t)}$ for $t \in\mathcal{T}$ \\
\textbf{Output}: An integer solution for COPRA\\
\caption{RA for constructing integer solutions}\label{alg:RA}
\begin{algorithmic}[1]
\FOR{$t\in\mathcal{T}$}\label{step1:start}
\WHILE{(exists $y_{ti}$ with fractional value)}
\STATE $y_{ti}\leftleftarrows 1$ if $y_{ti}=1$
\STATE $g=\max\{y_{ti}:0<y_{ti}<1\}$
\STATE $j=\arg\max\{y_{ti}:0<y_{ti}<1\}$
\STATE $\Phi \leftarrow L$~\textbf{if}~($t=1$)~\textbf{else}~ $\Phi\leftarrow S$
\STATE $y_{tj} \leftleftarrows 1$~\textbf{if}~($s_j+\underset{\substack{i\in\mathcal{I}\\ y_{ti}=1}}\sum s_i\le \Phi$)~\textbf{else}~$y_{tj} \leftleftarrows 0$
\STATE Solve SP$_2^{(t)}$
\ENDWHILE
\ENDFOR 
\STATE $\hat{\bm{y}}\leftarrow\{y_{ti}:t\in\mathcal{T},i\in\mathcal{I}\}$\label{step1:end}
\STATE Solve formulation~\eqref{repair} and obtain the values of $\bm{y}$ \label{step2}
\FOR{$t\in\mathcal{T}$}\label{step3:start}
\FOR{$i\in \mathcal{I}: y_{ti}=0$}
\STATE $\mathpzc{c}\leftarrow$~the~first~$N$~elements~in~$ \{(j,a):x_{tja}=1,  j\in \mathcal{R}_i, a\in\mathcal{A}_{tj}\}$
 with the highest relations to $i$ 
\ENDFOR
\ENDFOR \label{step3:end}
\end{algorithmic}
\end{algorithm}

\begin{algorithm}
\caption{The main steps of LDA}\label{alg:LDA}
\begin{algorithmic}[1]
\STATE Initialize $K$,~$\epsilon_1$,~and $\epsilon_2$\label{lda:init_0}
\STATE $\bm{\lambda}\leftarrow \bm{0}$,~$k\leftarrow 1$\label{lda:init_1}
\STATE $\text{LBD}\leftarrow 0$,~$\bar{w}\leftarrow \infty$\label{lda:init_2}
\REPEAT 
\STATE Solve SP$_1^{(i)}$ for $i \in \mathcal{I}$ and obtain $\bm{x}^{(k)}$\label{lda:sp1}
\STATE Solve SP$_2^{(t)}$ for $t \in \mathcal{T}$ by Algorithm~$\ref{alg:CGA}$ and obtain $\bm{x}'^{(k)}$\label{lda:sp2}
\STATE Calculate $L(\bm{\lambda}^{(k)})$  which is \eqref{ld_obj} \label{lda:lag}
\STATE \textbf{if}~$L(\bm{\lambda}^{(k)})>$LBD~\textbf{then}~LBD$\leftarrow \label{lda:lbd} L(\bm{\lambda}^{(k)})$
\STATE Apply Algorithm~$\ref{alg:RA}$ to obtain an integer solution and its  objective function value $U$\label{lda:repair}
\STATE \textbf{if}~$U<\bar{w}$~\textbf{then}~$\bar{w} \leftarrow U$\label{lda:u}
\STATE Calculate $\boldsymbol{\lambda}^{(k+1)}=\bm{\lambda}^{(k)}+t^{(k)}{\bm{d}}^{(k)}$ where $t^{(k)}=\eta\frac{\bar{w}-L(\bm{\lambda}^{(k)})}{{||\bm{d}^{(k)}}||^2}$,
$\bm{d}^{(k)}={\bm{x}}^{(k)}-{\bm{x}'}^{(k)}$\label{lda:update_lambda}
\STATE $k\leftarrow k+1$\label{lda:increase_k}
\UNTIL{$||\bm{d}^{(k)}||>\epsilon_1$ and $||{\bm{\lambda}^{(k+1)}-\bm{\lambda}^{(k)}||}>\epsilon_2$                                                                      and~$k>K$} \label{lda:stop}
\end{algorithmic}
\end{algorithm}

\subsection{Algorithm Summary}
The main steps of LDA is shown in Algorithm~\ref{alg:LDA}.
Line~\ref{lda:init_0} initialize the total number of iterations $K$ to perform, and tolerance parameters $\epsilon_1$ and $\epsilon_2$.
Lines~\ref{lda:init_1} and \ref{lda:init_2} initialize the vector of Lagrangian multiplier $\bm{\lambda}$ to $0$, the iteration counter $k=1$, the lower bound LBD to zero, and the best found solution $\bar{w}$ to $\infty$. Lines~\ref{lda:sp1} and \ref{lda:sp2} solve the SP$_1^{(i)}$ for $i\in\mathcal{I}$ and SP$_2^{(t)}$ for $t\in\mathcal{T}$, respectively. Lines~\ref{lda:lag} and \ref{lda:lbd} calculate the Lagrangian function value, and update the LBD if a higher lower bound is found. Lines~\ref{lda:repair} finds a solution for the problem, and then Line~\ref{lda:u} updates the current upper bound if a solution with lower objective function value is obtained. Line~\ref{lda:update_lambda} updates the Lagrange multipliers, and Line~\ref{lda:increase_k} increases the iteration counter by one. Finally, Line~\ref{lda:stop} checks whether a stopping criterion is met.

\section{Performance Results}
In this section, we present performance evaluation 
results of LDA and GA.
We first consider small-size problem instances, and evaluate 
the performances of LDA and GA by comparing them to the global optimum obtained from ILP~\eqref{ilp}.
We report the (relative) deviation from the optimum, referred to as the optimality gap.
For large-size problem instances, it is computationally difficult to
obtain global optimum. Instead, we use the LBD derived from
LDA as the reference value. This is a
valid comparison because the deviation with respect to the global
optimum will never exceed the deviation from the LBD. We will
see that, numerically, using the LBD remains accurate in evaluating
 optimality.

The content
popularity is modeled by a ZipF distribution,
i.e., the probability where the $i$-th content is requested
is $\frac{i^{-\gamma}}{\sum_{i\in\mathcal{I}}i^{-\gamma}}$ \cite{2013_kar,2020_ahani}. Here $\gamma$ is the shape parameter and it is set to $\gamma= 0.56$~\cite{2013_kar}. 
The sizes of content items are generated within interval $[1,10]$. We have set
the the cache capacity to $50\%$ of the total size of content items,
i.e., $C = 0.5\sum_{i\in\mathcal{I}}s_i$. The capacity of backhaul link is set to
$L = \rho\sum_{i\in\mathcal{I}}s_i$ where parameter $\rho$ steers the backhaul capacity in
relation to the total size of content items. 
The probability of accepting a related content is generated in interval $[0.6,1)$.
The maximum AoI that a content can take is set to two.
We use content-specific and time-specific functions including linear and nonlinear ones from the literature~\cite{2017_sun,2019_sun} to model
the AoI cost of content items. Specifically, for each content,
one of the following functions is randomly selected: $f_{tia}=1+\alpha_{ti}a$, $f_{tia}=\frac{1}{1-\alpha_{ti}a}$, and $f_{tia}=e^{\alpha_{ti}a}$.
The functions are made content-specific and time-specific by varying parameter $\alpha_{ti}$. We remark
that the performance of LDA remains largely the same if only one type of function is used
for all contents. The use of multiple functions is to show that the
algorithm works in general with diverse functions.
We will vary parameters
$I$, $T$, and $\rho$, and study their impact on the overall cost and algorithm performance.
For each input setup, we have generated 10 problem instances and we report the average cost.

Figure~\ref{item_small_recomVSnorecom} shows the total cost returned by LDA when recommendation is utilized, and LDA with no recommendation (denoted by LDC-NC).  The figure shows that, interestingly, the total cost decreases by more than $50\%$ with recommendation. Another interesting point is that the reduction is even more when the number of content items increases. From this result, the consideration of recommendation optimization is relevant.

\begin{figure}[ht]
\centering
  \includegraphics[scale=0.40]{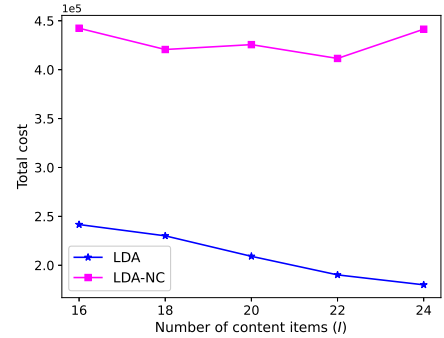}
  \caption{Impact of $I$ on total cost when $T=12$, $S=0.5\sum_{i\in\mathcal{I}}s_i$, $L=0.3\sum_{i\in\mathcal{I}}s_i$, and $\gamma=0.56$. The blue and pink lines show the total cost with and without recommendation, respectively.}\label{item_small_recomVSnorecom}
\end{figure}

Figures~\ref{impact:I_small}-\ref{impact:rho_small} and Figures~\ref{impact:I_large}-\ref{impact:rho_large} show the performance results for the small-size and large-size problem instances, respectively. In Figures~\ref{impact:I_small}-\ref{impact:rho_small}, the green line represents the global optimum computed using ILP~\eqref{ilp}. In Figures~\ref{impact:I_large}-\ref{impact:rho_large}, the black line represents the LBD obtained from LDA. In all figures, the blue and red lines represent the overall cost returned by LDA and GA, respectively.  The deviation from global optimum for LDA is within a few percent, while for GA it is significantly larger. Moreover, the results for both small-size and large-size problem instances are consistent.

Figure~\ref{impact:I_small} shows the impact of content items on the total cost for small-size problem instances. The overall cost slightly decreases with the number of contents. This is due to the fact that the capacity of cache is set relatively to the total. Namely, with larger number of contents, more capacity is available, and hence more opportunity to serve content requests from the cache. This effect, however, can not be seen for large problem instances due to a saturation effect, see Figure~\ref{impact:I_large}. As can be seen the cost has fluctuations due to instable solutions of GA. For small-size problems, the optimality gap of GA is about $57\%$, while for LDA it is about $7\%$ from global optimum. For large-size problems, the performance of LDA remains the same, while that of GA increases to $70\%$. Intuitively, the reason is that with larger number of items, the problem becomes too difficult for a simple algorithm such as GA.

Figures~\ref{impact:T_small} shows the impact of time slots for small-size problem instances. As can be seen, the cost increases with number of time slots. Apparently, this is because with more time slots, there are more requests to serve, and hence higher cost. GA has an optimality gap around $60\%$, while for LDA the gap is only $8\%$. The results for large-size problems are shown in Figure~\ref{impact:T_large}. LDA consistently shows good performance, whereas the results of GA are very sub-optimal. It is worth noting that the optimality gaps of both LDA and GA slightly increase with the number of time slots. 

Figure~\ref{impact:rho_small} shows the impact of $\rho$ on the total cost. Larger $\rho$ means
higher backhaul capacity. The costs of both LDA and GA decrease sharply when $\rho$ increases from $10\%$ to $20\%$, then the decrease slows down due to a saturation effect. The optimality gap of LDA is $17.5\%$ when $\rho=10\%$. This is because when the backhaul capacity is extremely limited, very few content items can be updated in a time slot, and as a result even one or two sub-optimal choices would largely impact the performance. When $\rho$ increases to $20\%$, the cost significantly decreases and the optimality gap decreases as well to $7.8\%$. For higher value of $\rho$, the gap slightly decreases further and stays around $7\%$. For GA the deviation from optimality is high no matter $\rho$ is small or not. Similar trends can be seen for large-size problems, see Figure~\ref{impact:rho_large}. Note that in the figure it may not be clear that the result of LDA and LBD both decrease with $\rho$. To show this, we have plotted a subfigure in the middle-right section of the figure.

\begin{figure}[t]
\centering
\includegraphics[scale=0.40]{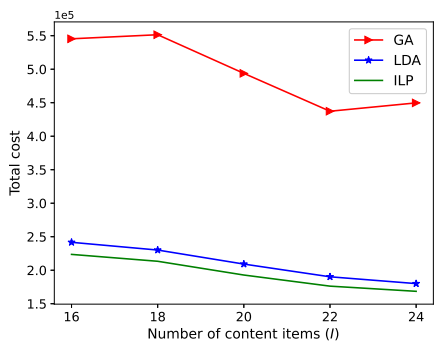}
\caption{Impact of $I$ on total cost when $T=6$, $S=0.5\sum_{i\in\mathcal{I}}s_i$, $L=0.3\sum_{i\in\mathcal{I}}s_i$, and $\gamma=0.56$.}\label{impact:I_small}
\end{figure}

\begin{figure}[t]
\centering
\includegraphics[scale=0.40]{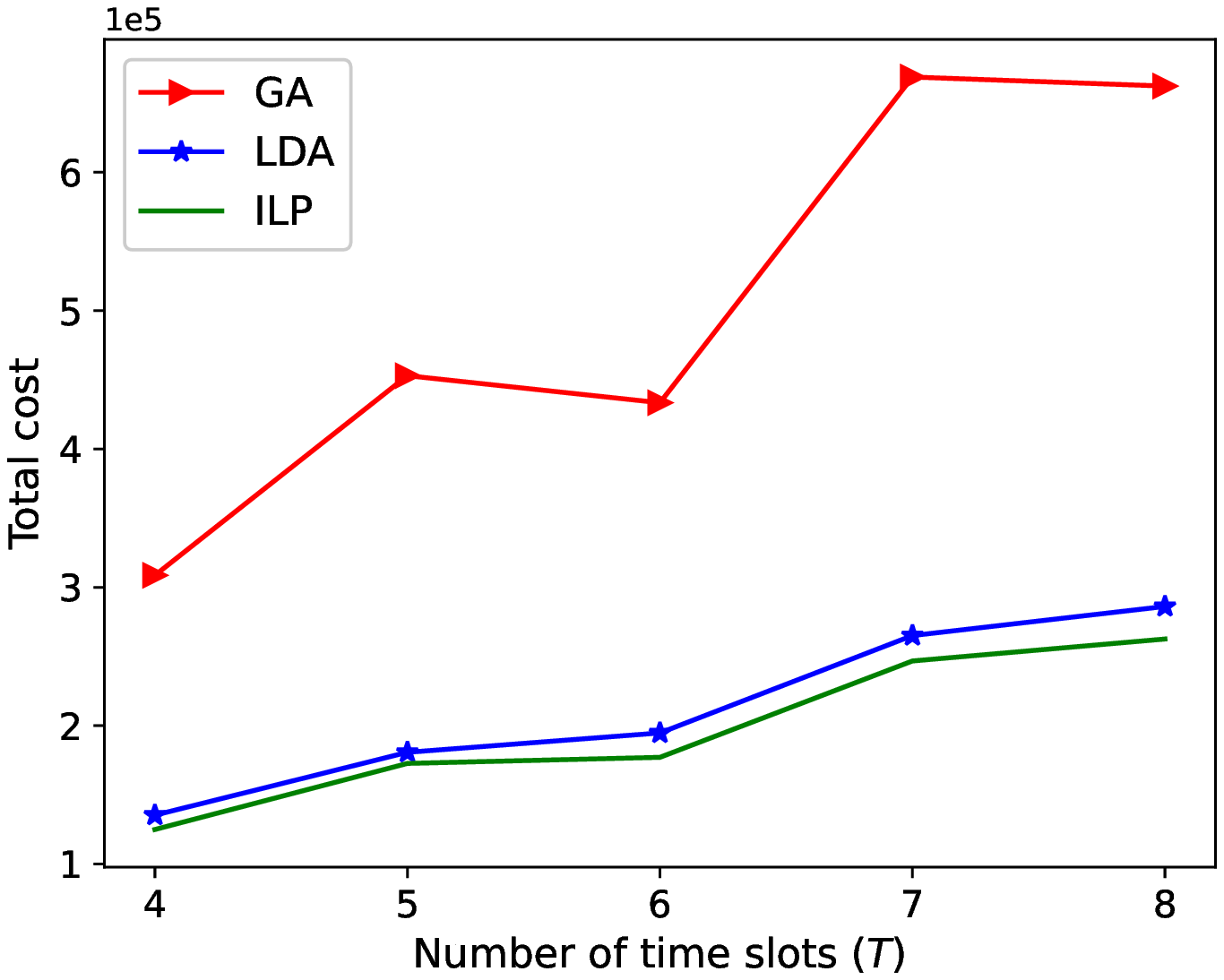}
\caption{Impact of $T$ on total cost when $I=20$, $S=0.5\sum_{i\in\mathcal{I}}s_i$, $L=0.3\sum_{i\in\mathcal{I}}s_i$, and $\gamma=0.56$.}\label{impact:T_small}
\end{figure}

\begin{figure}[t]
\centering
\includegraphics[scale=0.40]{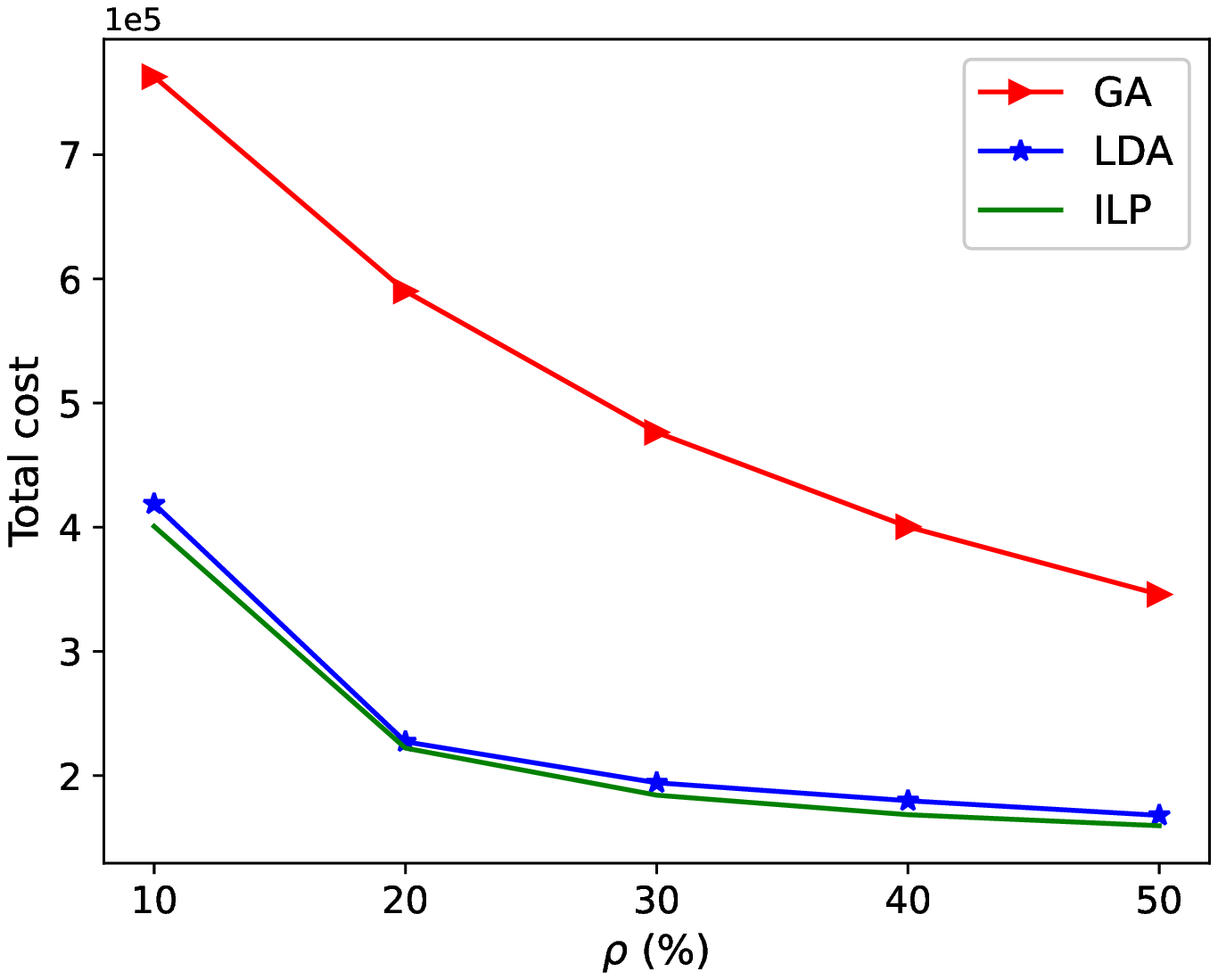}
\caption{Impact of $\rho$ on total cost when $I=20$, $T=6$, $S=0.5\sum_{i\in\mathcal{I}}s_i$, and $\gamma=0.56$.}\label{impact:rho_small}
\end{figure}

\begin{figure}[t]
\centering
\includegraphics[scale=0.40]{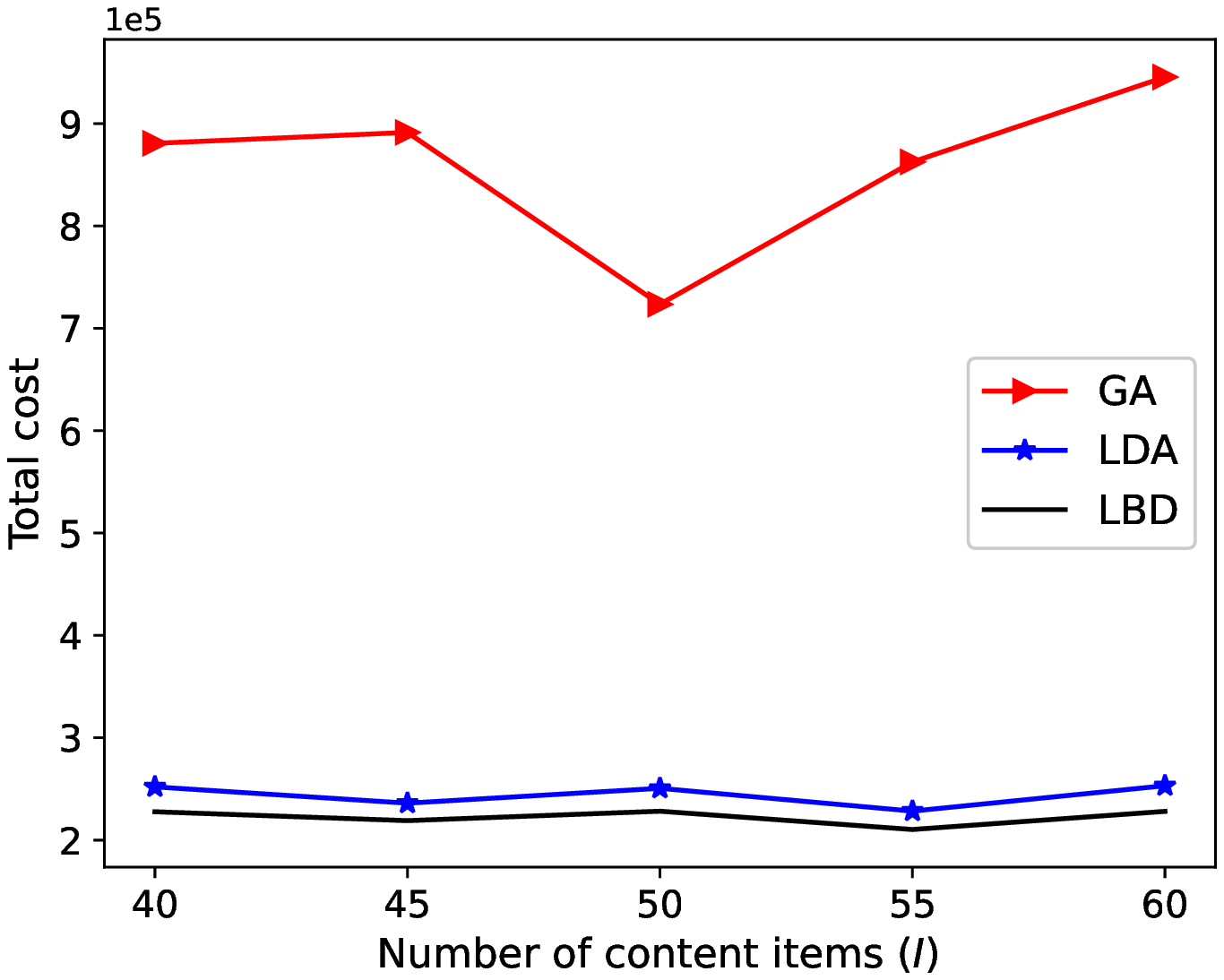}
\caption{Impact of $I$ on total cost when $T=12$, $S=0.5\sum_{i\in\mathcal{I}}s_i$, $L=0.3\sum_{i\in\mathcal{I}}s_i$, and $\gamma=0.56$.}\label{impact:I_large}
\end{figure}

\begin{figure}[t]
\centering
\includegraphics[scale=0.40]{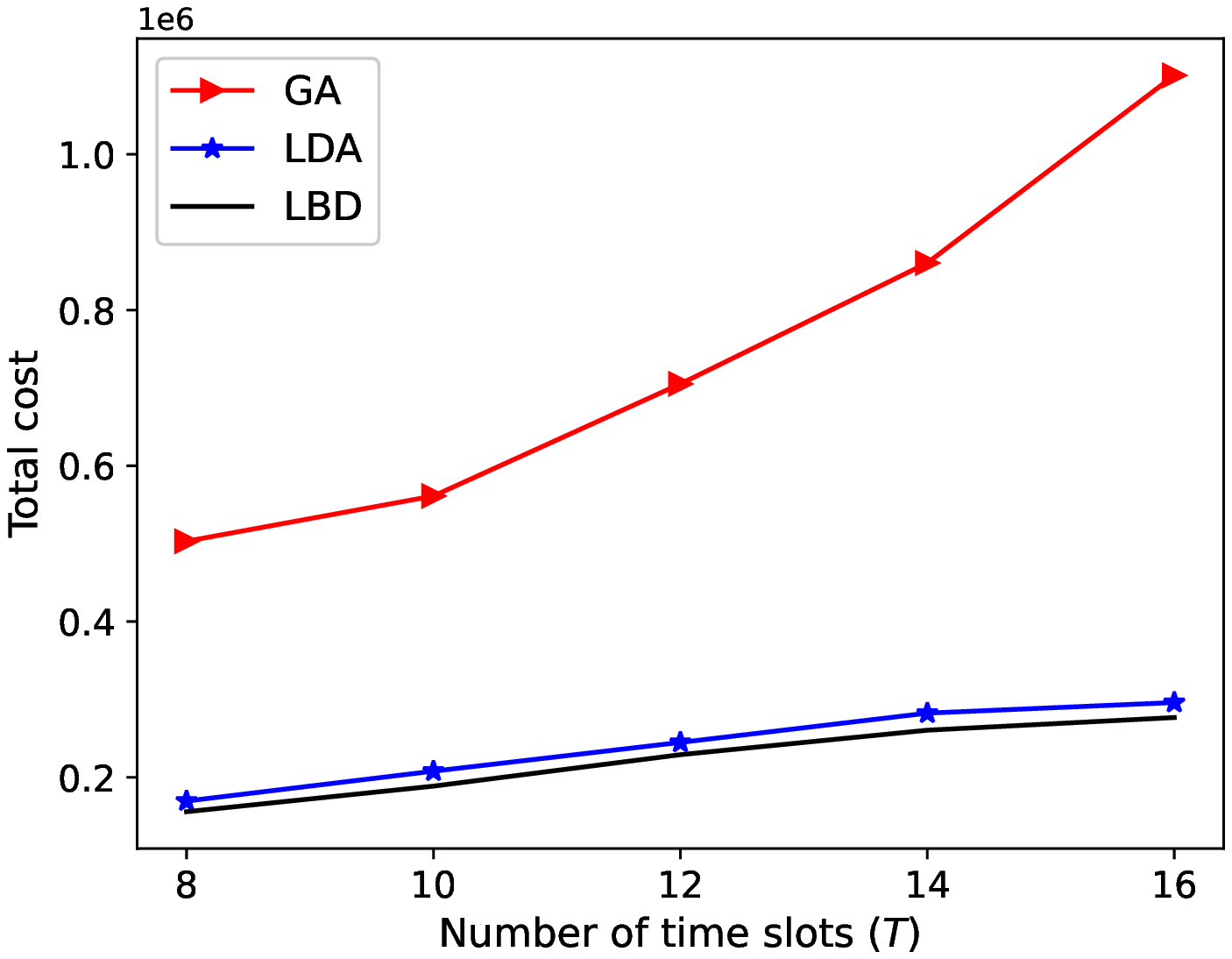}
\caption{Impact of $T$ on total cost when $I=50$, $S=0.5\sum_{i\in\mathcal{I}}s_i$, $L=0.3\sum_{i\in\mathcal{I}}s_i$, and $\gamma=0.56$.}\label{impact:T_large}
\end{figure}

\begin{figure}[t]
\centering
\includegraphics[scale=0.40]{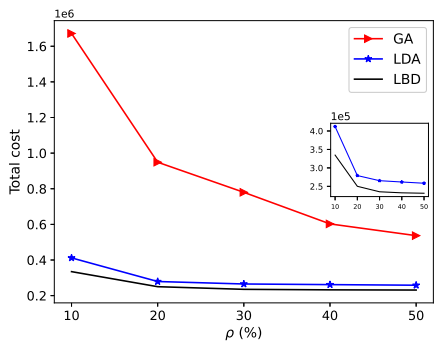}
\caption{Impact of $\rho$ on total cost when $I=50$, $T=12$, $S=0.5\sum_{i\in\mathcal{I}}s_i$, and $\gamma=0.56$.}\label{impact:rho_large}
\end{figure}

\section{Conclusions}
We have studied optimal scheduling of cache updates where  AoI of contents and recommendation are jointly taken into account. With both AoI and recommendation, the problem is hard even for one single time slot. 
We formulated the problem as an integer liner program (ILP). The ILP
provides optimal solutions, but it is not practical to large problem instances.
Simple algorithms are not likely to be effective, and this finding is obtained via the poor performance of a greedy algorithm (GA).
To arrive at good solutions efficiently, one has to analyze and exploit the structure of this optimization problem. We achieve this by the Lagrangian decomposition algorithm (LDA) that allows for decomposition for handling large-scale problem instances.
LDA decomposes the problem into several subproblems where each of them can be solved efficiently.
The algorithm provides solutions within a few percentage from global optimality. 

\bibliographystyle{IEEEtran}
\bibliography{IEEEabrv,ForIEEEBib}

\begin{IEEEbiography}[{\includegraphics[width=1in,height=1.25in,clip,keepaspectratio]{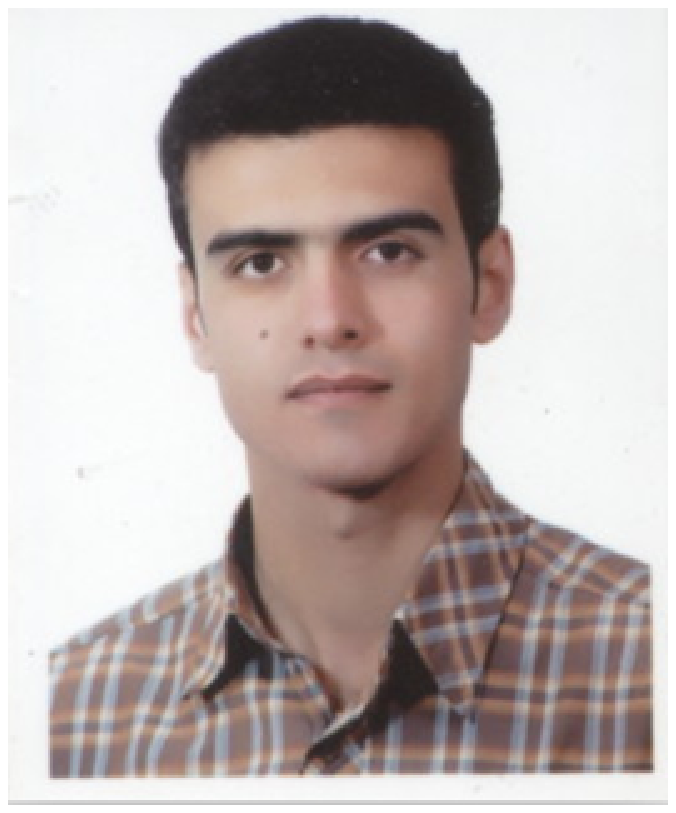}}]
{Ghafour Ahani} received the B.Sc. and the M.Sc. degrees in applied mathematics from the University of Kurdistan, Sanandaj, Iran, in 2008 and 2010 respectively. He is currently pursuing the Ph.D. degree with the IT department, Uppsala University, Sweden. His research interests include mathematical optimization applied to networks and resource allocation in communication systems.
\end{IEEEbiography}

\begin{IEEEbiography}[{\includegraphics[width=1in,height=1.25in,clip,keepaspectratio]{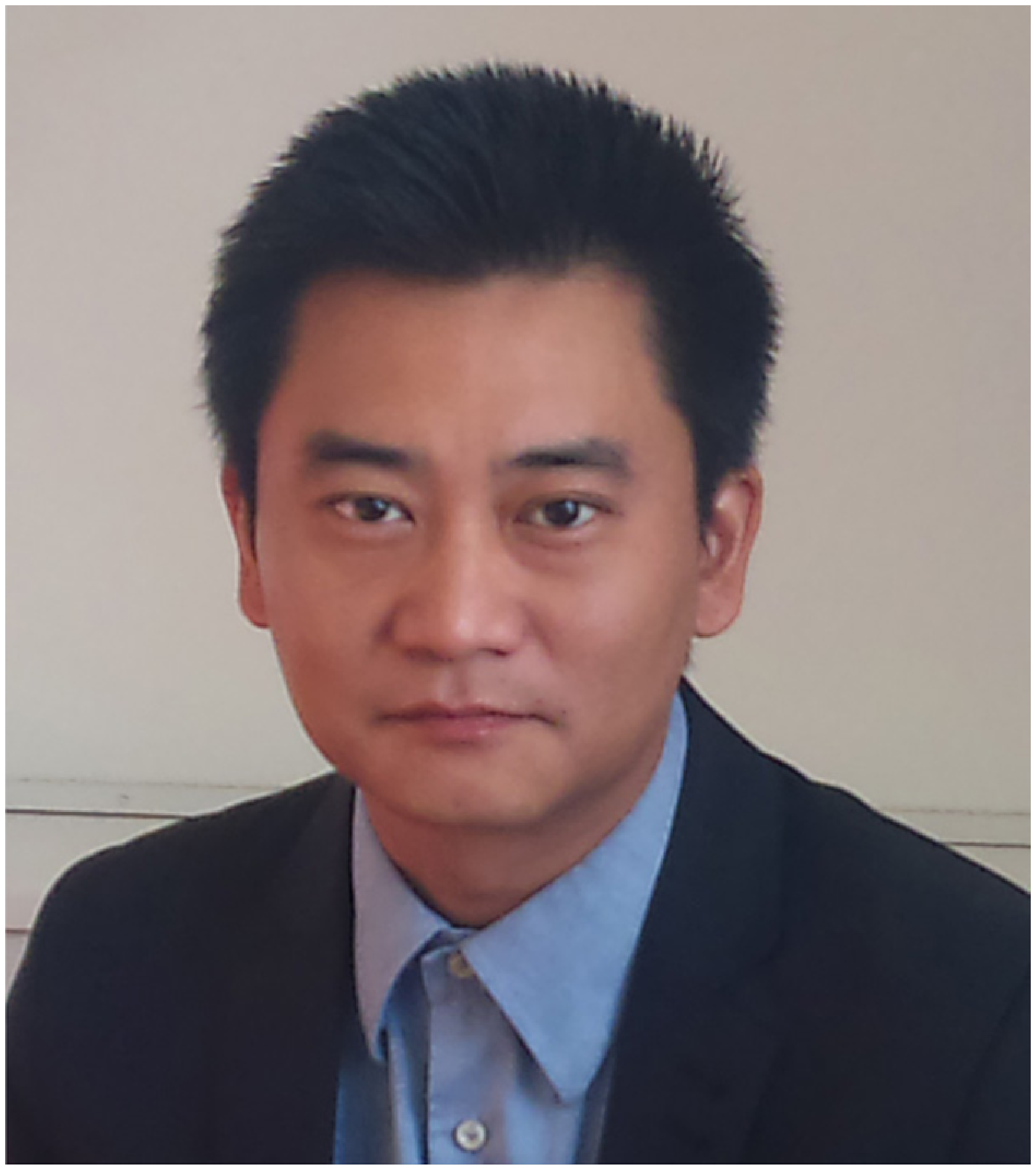}}]
{Di Yuan} (SM'16)  received his MSc degree in Computer Science and Engineering, and PhD degree in Optimization at Linköping Institute of Technology in 1996 and 2001, respectively. After his PhD, he has been associate professor and then full professor at the Department of Science and Technology, Linköping University, Sweden. In 2016 he joined Uppsala University, Sweden, as chair professor. His current research mainly addresses network optimization of 4G and 5G systems, and capacity optimization of wireless networks. Dr Yuan has been guest professor at the Technical University of Milan (Politecnico di Milano), Italy, in 2008, and senior visiting scientist at Ranplan Wireless Network Design Ltd, United Kingdom, in 2009 and 2012.  In 2011 and 2013 he has been part time with Ericsson Research, Sweden. In 2014 and 2015 he has been visiting professor at the University of Maryland, College Park, MD, USA. He is an area editor of the Computer Networks journal. He has been in the management committee of four European Cooperation in field of Scientific and Technical Research (COST) actions, invited lecturer of European Network of Excellence EuroNF, and Principal Investigator of several European FP7 and Horizon 2020 projects. He is a co-recipient of IEEE ICC’12 Best Paper Award, and supervisor of the Best Student Journal Paper Award by the IEEE Sweden Joint VT-COM-IT Chapter in 2014.
\end{IEEEbiography}
\end{document}